%% file: elicitation.tex
\newcommand{\longversion}[1]{#1}
\newcommand{\shortversion}[1]{}

\longversion{
\documentclass[a4paper,11pt]{article}
\author{Palash Dey\\ \texttt{palash@tifr.res.in}\\\\ Tata Institute of Fundamental Research - Mumbai, India.}
}

\shortversion{\relax
\documentclass[letterpaper]{article} 
\usepackage{aaai18}  
\usepackage{times}  
\usepackage{helvet}  
\usepackage{courier}  
\usepackage{url}  
\usepackage{graphicx}  
\frenchspacing  
\setlength{\pdfpagewidth}{8.5in}  
\setlength{\pdfpageheight}{11in}  
  \pdfinfo{
/Title (On Testing Societal Preference Domains)
/Author (Paper id:)}
\setcounter{secnumdepth}{3}
}

\input{preamble.tex}

\title{{\bf Manipulative Elicitation -- A New Attack on Elections with Incomplete Preferences}}
\author{Palash Dey\\ \texttt{palash@tifr.res.in}\\ Tata Institute of Fundamental Research, Mumbai}

\begin{document}

\maketitle

\input{abstract.tex}
\input{introduction.tex}

\input{prelim.tex}
\input{results.tex}
\input{conclusion.tex}


\longversion{\bibliographystyle{alpha}}
\shortversion{\bibliographystyle{named}}

\bibliography{references}

\end{document}

%% file: preamble.tex
\usepackage{amssymb,amsmath,amsthm}
\usepackage{makecell}
\usepackage{multirow}
\usepackage{mathtools}
\usepackage{thmtools,thm-restate}
\longversion{\usepackage{datetime}
\usepackage{fullpage}}
\usepackage{relsize}
\usepackage[usenames,dvipsnames]{color}
\usepackage{xcolor}
\usepackage[colorlinks=true,linkcolor=blue,citecolor=Mahogany]{hyperref}
\usepackage{eulervm}
\usepackage{enumitem}
\shortversion{
\floatsep 0pt plus 1pt minus 1pt
\textfloatsep 0pt plus 1pt minus 1pt
\intextsep 0pt plus 1pt minus 1pt
\dblfloatsep 0pt plus 1pt minus 1pt
\dbltextfloatsep 0pt plus 1pt minus 1pt

\usepackage{titlesec}
\titlespacing{\section}{0pt}{1pt}{1pt}
\titlespacing{\subsection}{0pt}{2pt}{1.5pt}
\setlength{\belowdisplayskip}{2pt} \setlength{\belowdisplayshortskip}{0pt}
\setlength{\abovedisplayskip}{2pt} \setlength{\abovedisplayshortskip}{0pt}
}

\longversion{\usepackage{charter}}

\usepackage{xpatch}
\makeatletter
\xpatchcmd{\@endpart}{\vfil\newpage}{}{}{}
\xpatchcmd{\@endpart}{\newpage}{}{}{}
\makeatother

\usepackage{tikz}
\usetikzlibrary{shapes.geometric, arrows}

\usepackage{algpseudocode,algorithm,algorithmicx}

\algrenewcommand\algorithmicrequire{\textbf{Input:}}
\algrenewcommand\algorithmicensure{\textbf{Output:}}
\algnewcommand{\Initialize}[1]{%
  \State \textbf{Initialize:}
  \Statex \hspace*{\algorithmicindent}\parbox[t]{.8\linewidth}{\raggedright #1}
}
\algdef{SE}[FUNCTION]{Function}{EndFunction}%
   [2]{\algorithmicfunction\ \textproc{#1}\ifthenelse{\equal{#2}{}}{}{(#2)}}%
   {\algorithmicend\ \algorithmicfunction}%

\usepackage{url}

\makeatletter
\newcommand{\mathleft}{\@fleqntrue\@mathmargin\parindent}
\newcommand{\mathcenter}{\@fleqnfalse}
\makeatother

\usepackage{xspace}

\newcommand{\el}{\ensuremath{\ell}\xspace}
\newcommand{\suc}{\ensuremath{\succ}\xspace}

\DeclareMathOperator*{\argmin}{arg\!min}
\DeclareMathOperator*{\argmax}{arg\!max}

\renewcommand{\ge}{\geqslant}
\renewcommand{\le}{\leqslant}

\newcommand{\YES}{\textsc{Yes}\xspace}
\newcommand{\NO}{\textsc{No}\xspace}

\newcommand{\pr}{\ensuremath{\prime}\xspace}
\newcommand{\prr}{{\ensuremath{\prime\prime}}\xspace}

\newcommand{\SE}{\textsc{Manipulative Elicitation}\xspace}

\newcommand{\SECPL}{\textsc{Manipulative Elicitation with Candidate Pair Limit}\xspace}

\newcommand{\NP}{\ensuremath{\mathsf{NP}}\xspace}

\newcommand{\coNPH}{\ensuremath{\mathsf{co}}-\ensuremath{\mathsf{NP}}-hard\xspace}
\newcommand{\NPH}{\ensuremath{\mathsf{NP}}-hard\xspace}
\newcommand{\NPC}{\ensuremath{\mathsf{NP}}-complete\xspace}

\newcommand{\XTC}{X3C\xspace}
\newcommand{\SCFT}{\textsc{Set Cover Frequency Two}\xspace}
\newcommand{\VC}{\textsc{Vertex Cover}\xspace}

\newcommand{\NB}{\ensuremath{\mathbb N}\xspace}

\renewcommand{\AA}{\ensuremath{\mathcal A}\xspace}

\newcommand{\CC}{\ensuremath{\mathcal C}\xspace}
\newcommand{\DD}{\ensuremath{\mathcal D}\xspace}
\newcommand{\EE}{\ensuremath{\mathcal E}\xspace}
\newcommand{\FF}{\ensuremath{\mathcal F}\xspace}
\newcommand{\GG}{\ensuremath{\mathcal G}\xspace}

\newcommand{\LL}{\ensuremath{\mathcal L}\xspace}

\newcommand{\NN}{\ensuremath{\mathcal N}\xspace}

\newcommand{\PP}{\ensuremath{\mathcal P}\xspace}
\newcommand{\QQ}{\ensuremath{\mathcal Q}\xspace}
\newcommand{\RR}{\ensuremath{\mathcal R}\xspace}
\renewcommand{\SS}{\ensuremath{\mathcal S}\xspace}

\newcommand{\UU}{\ensuremath{\mathcal U}\xspace}
\newcommand{\VV}{\ensuremath{\mathcal V}\xspace}

\renewcommand{\lll}{\ensuremath{\mathfrak f}\xspace}

\newcommand{\ppp}{\ensuremath{\mathfrak p}\xspace}

\newcommand{\sss}{\ensuremath{\mathfrak s}\xspace}

\usepackage{nicefrac}

\newcommand{\nfrac}{\nicefrac}

\newcommand{\eps}{\varepsilon}
\renewcommand{\epsilon}{\eps}

\usepackage{cleveref}

\newtheorem{theorem}{\bf Theorem}
\newtheorem{lemma}{\bf Lemma}
\newtheorem{claim}{\bf Claim}
\newtheorem{corollary}{\bf Corollary}
\newtheorem{definition}{\bf Definition}

\crefname{example}{Example}{Example}
\crefname{theorem}{Theorem}{Theorem}
\crefname{observation}{Observation}{Observation}
\crefname{lemma}{Lemma}{Lemma}
\crefname{corollary}{Corollary}{Corollary}
\crefname{proposition}{Proposition}{Proposition}
\crefname{definition}{Definition}{Definition}
\crefname{claim}{Claim}{Claim}
\crefname{reductionrule}{Reduction rule}{Reduction rule}
\crefname{ineq}{inequality}{Inequality}

\longversion{
\numberwithin{example}{section}
\numberwithin{proposition}{section}
\numberwithin{claim}{section}
\numberwithin{theorem}{section}
\numberwithin{corollary}{section}
\numberwithin{observation}{section}
\numberwithin{lemma}{section}
\numberwithin{definition}{section}
\numberwithin{reductionrule}{section}
}

%% file: abstract.tex
\begin{abstract}
Lu and Boutilier~\longversion{\cite{LuB11a}}\shortversion{\shortcite{LuB11a}} proposed a novel approach based on ``minimax regret'' to use classical score based voting rules in the setting where preferences can be any partial (instead of complete) orders over the set of alternatives. We show here that such an approach is vulnerable to a new kind of manipulation which was not present in the classical (where preferences are complete orders) world of voting. We call this attack ``manipulative elicitation.'' More specifically, it may be possible to (partially) elicit the preferences of the agents in a way that makes some distinguished alternative win the election who may not be a winner if we elicit every preference completely. More alarmingly, we show that the related computational task is polynomial time solvable for a large class of voting  rules which includes all scoring rules, maximin, Copeland$^\alpha$ for every $\alpha\in[0,1]$, simplified Bucklin voting rules, etc. We then show that introducing a parameter per pair of alternatives which specifies the minimum number of partial preferences where this pair of alternatives must be comparable makes the related computational task of manipulative elicitation \NPC for all common voting rules including a class of scoring rules which includes the plurality, $k$-approval, $k$-veto, veto, and Borda voting rules, maximin, Copeland$^\alpha$ for every $\alpha\in[0,1]$, and simplified Bucklin voting rules. Hence, in this work, we discover a fundamental vulnerability in using minimax regret based approach in partial preferential setting and propose a novel way to tackle it.
\end{abstract}

%% file: introduction.tex
\section{Introduction}

Aggregating preferences of a set of agents over a set of alternatives is a fundamental problem in voting theory which has been used in many applications in AI for making various decisions. Prominent examples of such applications include collaborative filtering~\cite{PennockHG00}, similarity search~\cite{Fagin}, winner determination in sports competitions~\cite{BetzlerBN14}, etc.~\cite{moulin2016handbook}. In a typical scenario of voting, we have a set of alternatives, a tuple of ``preferences'', called a profile, over the set of alternatives, and a voting rule which chooses a set of alternatives as winners based on the profile. Classically, preferences are often modeled as complete orders over the set of alternatives. However, in typical applications of voting in AI, collaborative filtering for example, the number of alternatives is huge and we have only partial orders over the set of alternatives as preferences.

There have been many attempts to extend the use of voting theory in settings with incomplete preferences. The approach of Konczak and Lang~\longversion{\cite{konczak2005voting}}\shortversion{\shortcite{konczak2005voting}} was to study the possible and necessary winner problems. In these problems, the input is a profile of partial preferences and we want to compute the set of alternatives who wins (under some fixed voting rule) in at least one completion of the profile for the possible winner problem; for the necessary winner problem, we want to compute the set of alternatives who wins in every completion of the profile. There have been substantial research effort in the last decade to better understand these two problems~\cite{lang2007winner,pini2007incompleteness,walsh2007uncertainty,xia2008determining,betzler2009multivariate,chevaleyre2010possible,betzler2010partial,baumeister2011computational,lang2012winner,faliszewski2014complexity,journalsDeyMN16,deypartial,frugaljournal,DeyMN15a,deymfcs2017}. One of the main criticisms of this approach is that the definition of a necessary winner is so strong that none of the alternatives may satisfy it whereas the definition of a possible winner is so relaxed that a large number of alternatives may satisfy it. Moreover, the computational problem of finding the set of possible winners is \NPH for most of the common voting rules (finding the set of necessary winners is also \coNPH for some voting rules, ranked pairs for example)~\cite{xia2008determining}.

Lu and Boutilier~\longversion{\cite{LuB11a}}\shortversion{\shortcite{LuB11a}} took a completely different approach to handle incomplete preferences and proposed a worst case regret based approach for score based voting rules. These voting rules assign some score to every alternative based on the profile and select the alternatives with the maximum (or minimum) score as winners. Many popular voting rules, for example, scoring rules, maximin, Copeland, etc. are score based voting rules. For score based voting rules, intuitively speaking, the worst case regret, called maximum regret in~\cite{LuB11a}, of declaring an alternative $w$ as a winner is the maximum possible difference between the score of $w$ and the score of a winning alternative in any completion of the input partial profile; the winners of a partial profile are the set of alternatives with the minimum maximum (called minimax) regret. A completion of a partial profile is another profile where every preference is complete and it respects the orderings of the corresponding preference in the partial profile. The minimax regret based approach is not only theoretically robust as argued in~\cite{LuB11a} but also practically appealing since computing winners is polynomial time solvable for all commonly used voting rules.

\subsection{Motivation} Although the minimax regret based approach enjoys many exciting features, it introduces a new (which was not present in the classical setting with complete preferences) kind of attack on the election which we call ``manipulative elicitation.'' That is, it may be possible to partially elicit the preferences in such a way that makes some favorable alternative win the election. For example, let us consider a plurality election \EE where an alternative, say $w$, is the top alternative of one preference and another alternative, say $x$, is the top alternative of every other preference. In a plurality election, the winners are the set of alternatives who appear as the top alternative in the largest number of preferences. Hence, $x$ is the unique winner in \EE. Let us now consider a partial profile where, in every partial preference, only $w$ and every other alternative who is preferred less than $w$ in the corresponding preference in \EE are comparable. Let us call the resulting partial profile $\EE^\pr$. If $n$ is the number of preferences, then the minimax regret plurality score of $w$ in $\EE^\pr$ is $(n-1)$ whereas the minimax regret plurality score of every other alternative is $n$ which makes $w$ the unique winner of $\EE^\pr$. We call this phenomenon manipulative elicitation. The problem of manipulative elicitation is even more alarming in AI since, in many applications (collaborative filtering for example), the parts of the preferences that will be elicited can often be influenced and controlled in such settings.

\subsection{Our Contribution} Our main contribution in this paper is the discovery of the manipulative elicitation attack in regret based partial preferential setting. We also show that the corresponding computational problem for manipulative elicitation is polynomial time solvable for every {\em monotone} voting rule which includes all commonly used score based voting rules~[\Cref{thm:all_poly,cor:all_poly}]. Intuitively speaking, we call a score based voting rule monotone if improving the position of some alternative in any (complete) preference can only improve its score; we defer its formal definition till \Cref{sec:prelim}. To counter the negative result of \Cref{thm:all_poly}, we introduce a parameter per pair of alternatives which specifies the minimum number of partial preferences where these two alternatives should be comparable. We establish success of our approach by showing that the new constraints make the corresponding computational task of manipulative elicitation \NPC for a large class of scoring rules [\Cref{thm:secpl_sc}] which includes the plurality [\Cref{thm:secpl_plurality_2app}], veto [\Cref{thm:secpl_veto}], $k$-approval for any $k$, and Borda voting rules [\Cref{cor:kapp_borda}], maximin [\Cref{thm:secpl_maximin}], Copeland$^\alpha$ for every $\alpha\in[0,1]$ [\Cref{thm:secpl_copeland}], and simplified Bucklin [\Cref{thm:secpl_bucklin}] voting rules. We remark that there could be various ways to enforce lower bounds on the number of partial preferences where a particular pair of alternatives is comparable. For example, this can be a feature in the applications which would allow users to generate these bounds from some distribution which would in turn overrule the possibility of such manipulation (due to our hardness results).

%% file: prelim.tex
\section{Preliminaries and Problem Formulation}\label{sec:prelim}

For a positive integer $k$, we denote the set $\{1, 2, \ldots, k\}$ by $[k]$. Let $\AA=\{a_i: i\in[m]\}$ be a set of $m$ {\em alternatives}. We denote the set of all subsets of \AA of cardinality $2$ by ${\AA\choose 2}$. A complete order over the set \AA of alternatives is called a (complete) {\em preference}. We say that an alternative $a\in\AA$ is placed at the $\el^{th}$ position (from left or from top) in a preference \suc if $|\{b\in\AA: b\suc a\}|=\el-1$. We denote the set of all possible preferences over \AA by $\LL(\AA)$. A tuple $\suc=(\suc_i)_{i\in[n]}\in\LL(\AA)^n$ of $n$ preferences is called a {\em profile}. An {\em election} \EE is a tuple $(\suc,\AA)$ where \suc is a profile over a set \AA of alternatives. If not mentioned otherwise, we denote the number of alternatives and the number of preferences by $m$ and $n$ respectively. A map $r_c:\uplus_{n,|\mathcal{A}|\in\mathbb{N}^+}\mathcal{L(A)}^n\longrightarrow 2^\mathcal{A}\setminus\{\emptyset\}$ is called a \emph{voting rule}. Given an election $\EE$, we can construct from \EE a directed weighted graph $\GG_\EE$ which is called the \textit{weighted majority graph} of \EE. The set of vertices in $\GG_\EE$ is the set of alternatives in $\EE$. For any two alternatives $x$ and $y$, the weight of the edge $(x,y)$ is $\DD_\EE(x,y) = \NN_\EE(x,y) - \NN_\EE(y,x)$, where $\NN_\EE(a, b)$ is the number of preferences where the alternative $a$ is preferred over the alternative $b$ for $a,b\in\AA, a\ne b$. Examples of some common voting rules are as follows.

\begin{itemize}[leftmargin=0cm,itemindent=0.3cm,labelwidth=\itemindent,labelsep=0cm,align=left,noitemsep,topsep=2pt]
 \item \textbf{Positional scoring rules:} An $m$-dimensional vector $\alpha=\left(\alpha_1,\alpha_2,\dots,\alpha_m\right)\in\mathbb{N}^m$ with $\alpha_1\ge\alpha_2\ge\dots\ge\alpha_m$ and $\alpha_1>\alpha_m$ for every $m\in \mathbb{N}$ naturally defines a voting rule --- an alternative gets score $\alpha_i$ from a preference if it is placed at the $i^{th}$ position, and the score of an alternative is the sum of the scores it receives from all the preferences. The winners are the alternatives with the maximum score. Scoring rules remain unchanged if we multiply every $\alpha_i$ by any constant $\lambda>0$ and/or add any constant $\mu$. Hence, we can assume without loss of generality that for any score vector $\alpha$, we have $gcd((\alpha_i)_{i\in[m]})=1$ and there exists a $j< m$ such that $\alpha_\el = 0$ for all $\el>j$. We call such an $\alpha$ a {\em normalized} score vector. 
 If $\alpha_i$ is $1$ for $i\in [k]$ and $0$ otherwise, then, we get the {\em $k$-approval} voting rule. The {\em $k$-approval} voting rule is also called the {\em $(m-k)$-veto} voting rule. The $1$-approval voting rule is called the {\em plurality} voting rule and the $1$-veto voting rule is called the {\em veto} voting rule. If $\alpha_i=m-i$ for every $i\in[m]$, then we get the {\em Borda} voting rule.
 
 \item \textbf{Maximin:} The maximin score of an alternative $x$ is $\min_{y\ne x} \NN_\EE(x,y)$. The winners are the alternatives with the maximum maximin score.
 
 \item \textbf{Copeland$^{\alpha}$:} Given $\alpha\in[0,1]$, the Copeland$^{\alpha}$ score of an alternative $x$ is $|\{y\ne x:\DD_\EE(x,y)>0\}|+\alpha|\{y\ne x:\DD_\EE(x,y)=0\}|$. The winners are the alternatives with the maximum Copeland$^{\alpha}$ score. 
 
 \item \textbf{Simplified Bucklin:} An alternative $x$'s simplified Bucklin score is the minimum number $\ell$ such that $x$ is placed within the top $\ell$ positions in more than half of the preferences. The winners are the alternatives with the lowest simplified Bucklin score.
\end{itemize}

We call a voting rule {\em ``score based''} if the voting rule assigns some score to every alternative based on the profile and chooses either the set of alternatives with the maximum score or the set of alternatives with the minimum score as winners. All the above mentioned voting rules are score based\longversion{ --- positional scoring rules, maximin, and Copeland$^\alpha$ for every $\alpha\in[0,1]$ select the set of alternatives with the maximum score as winners and the simplified Bucklin voting rule selects the set of alternatives with the minimum score as winners}. \longversion{We now define few properties of a score based voting rule which will be relevant to us.}
We say that a score based voting rule $s$ is {\em monotone} if, for every positive integer $n$, every two profiles $(\suc_i)_{i\in[n]}$ and $(\suc_i^\pr)_{i\in[n]}$ over any finite set \AA of alternatives, and every alternative $x\in\AA$ such that $\{y\in\AA: x\suc_i y\} \subseteq \{y\in\AA: x\suc_i^\pr y\}$ for every $i\in[n]$, we have $s(x, (\suc_i)_{i\in[n]}) \le s(x, (\suc_i^\pr)_{i\in[n]})$. 
We call a voting rule $r$ {\em neutral} if, for every positive integer $n$, every profile $(\suc_i)_{i\in[n]}$ over any finite set $\AA = \{x_i: i\in[m]\}$ of $m$ alternatives, and every permutation $\sigma$ of $[m]$, we have $\sigma(r((\suc_i)_{i\in[n]})) = r((\sigma(\suc_i))_{i\in[n]})$ where \longversion{$\sigma(r((\suc_i)_{i\in[n]}))$ is the image of $r((\suc_i)_{i\in[n]})$ under $\sigma$ and }$\sigma(\suc_i) = x_{\sigma(1)}\suc x_{\sigma(2)}\suc \cdots\suc x_{\sigma(m)}$ if $\suc_i = x_1 \suc x_2 \suc \cdots \suc x_m$. We call a voting rule {\em worst efficient} if the worst possible score with $n$ preferences over $m$ alternatives can be computed in a polynomial (in $m$ and $n$) amount of time. We observe that all the voting rules mentioned above are neutral, worst efficient, and monotone if, for the case of simplified Bucklin voting rule, we replace the simplified Bucklin score with negative of that and choose the alternative with the maximum score.

\subsection{Incomplete Election and Minimax Regret Extension of Score Based Voting Rules}

Although preferences are classically modeled as complete orders, in many scenarios, preferences can be any partial order over \AA. We often denote a partial order \RR by the set $\{(a,b): a,b\in\AA, a\RR b\}$. Given a profile \PP of partial preferences (which we call a partial profile), we denote the set of all completions of \PP to complete orders by $\CC(\PP)$. Lu and Boutilier~\longversion{\cite{LuB11a}}\shortversion{\shortcite{LuB11a}} proposed a novel approach to extend the use of score based voting rules for settings with partial profiles based on a notion of regret. Let $s$ be a score based voting rule so that the winner is an alternative with the maximum score. Positional scoring rules, maximin, Copeland$^\alpha$ for every $\alpha\in[0,1]$, etc. are prominent examples of such score based voting rules. Let us denote the score that a score based voting rule $s$ assigns to an alternative $a\in\AA$ in a profile $\suc\in\LL(\AA)^n$ by $s(a,\suc)$. We denote the minimax regret voting rule based on a voting rule $s$ by $\sss$. For a profile \suc, let $s(\suc)=\argmax_{a\in\AA} \{ s(a, \suc)\}$. Given a partial profile \PP and a score based rule $s$, $\sss(\PP)$ is defined as follows.
\begin{align*}
s-Regret(a, \suc) &= |s(s(\suc),\suc) - s(a, \suc)|\\
s-MR(a, \PP) &= \max_{\suc\in\CC(\PP)} s-Regret(a, \suc)\\
\sss(\PP) &= \argmin_{a\in\AA} s-MR(a, \PP)
\end{align*}

\longversion{We remark that Lu and Boutilier~\cite{LuB11a} defined $s-Regret(a, \suc)$ without the absolute operator on the right; we choose to do so to take care of the simplified Bucklin voting rule.} For a partial profile \PP and a minimax regret (MR for short) voting rule \sss, we say that an alternative $a\in\AA$ co-wins if $a\in\sss(\PP)$ and wins uniquely if $\sss(\PP)=\{a\}$. 
For an alternative $a\in\AA$, if $s-MR(a, \PP) = s-Regret(a, \suc)$ for some $\suc\in\CC(\PP)$, then we call an alternative in $s(\suc)$ a {\em competing alternative} of $a$ in \PP.

We now formally define manipulative elicitation and the basic computational problem of manipulative elicitation for a score based voting rule $s$.

\begin{definition}[\sss-manipulative elicitation]
 For a profile \suc over a set \AA of alternatives, we say that a partial profile \PP is called a manipulative elicitation if $\suc\in\CC(\PP)$ and $\sss(\PP)=\{c\}$.
\end{definition}

\begin{definition}[\sss-\SE]\label{def:naive}
 Given a set \AA of alternatives, a profile $\suc\in\LL(\AA)^n$ of $n$ preferences, and an alternative $c\in\AA$, compute if there exists a partial profile \PP such that $\suc\in\CC(\PP)$ and $\sss(\PP)=\{c\}$?
\end{definition}

We will see in \Cref{thm:all_poly} that the \sss-\SE problem is polynomial time solvable for every neutral, monotone, and worst efficient score based voting rule. This shows that all the commonly used voting rule considered here are vulnerable under manipulative elicitation. In the hope to counter this drawback, we extend the basic problem in \Cref{def:naive} to \SECPL in \Cref{def:secpl}. We will indeed see in \Cref{sec:results_hard} that the \SECPL problem is \NPC for all the voting rules that we consider in this paper. For a partial profile $\suc=(\suc_i)_{i\in[n]}$ and $\{a,b\}\in{\AA\choose 2}$, we denote the number of partial preferences in \suc where $a$ and $b$ are comparable by $\ppp_{\{a,b\}}(\suc)$.


\begin{definition}[\sss-\SECPL]\label{def:secpl}
 Given a set \AA of alternatives, a profile $\suc\in\LL(\AA)^n$ of $n$ voters, a function $\lll: {\AA\choose 2} \longrightarrow \NB$ such that $0\le\lll(\{a,b\})\le n$ for every $\{a,b\}\in{\AA\choose 2}$, and an alternative $x\in\AA$, compute if there exists a partial profile \PP such that $\suc\in\CC(\PP),\lll(\{a,b\})\le \ppp_{\{a,b\}}(\suc)$ for every $\{a,b\}\in{\AA\choose 2}$ and $\sss(\PP)=\{x\}$?
\end{definition}


We remark that both the computational problems in \Cref{def:naive,def:secpl} have been defined for the unique winner case; we could as well define these problems for the co-winner case also. It turns out that all our proofs (except \Cref{thm:all_poly}) can be easily modified for the co-winner counterpart and our choice for defining these problems in the unique winner setting is only a matter of exposition.

%% file: results.tex
\section{Polynomial Time Algorithm for \SE}

Our first result is \Cref{thm:all_poly} which shows that the \SE problem is polynomial time solvable for a large class of voting rules.
\begin{theorem}\label{thm:all_poly}
 The \SE problem is polynomial time solvable for every monotone, neutral, and worst efficient score based voting rule $s$.
\end{theorem}

\begin{proof}
 Let $(\AA,\suc=(\suc_i)_{i\in[n]},c)$ be an arbitrary instance of \sss-\SE. \longversion{We first observe that every instance of \sss-\SE for the co-winner case is a \YES instance since, due to neutrality of $s$, $c\in\AA=\sss((\emptyset)^n)$ and $\suc\in\LL(\AA)^n=\CC((\emptyset)^n)$. So, let us consider the unique winner case. } Our algorithm is as follows. If $c$ receives the worst possible score in \suc, then we output \NO; otherwise we output \YES. Our algorithm runs in polynomial time since $s$ is worst efficient. To prove the correctness of our algorithm, we begin with \Cref{clm:no} below.
 
 \begin{claim}\label{clm:no}
 If the score of $c$ in \suc is the worst possible score (say $\beta_n$) that any alternative in \AA can possibly receive in any profile with $n$ preferences under\longversion{ the voting rule} $s$, then the \sss-\SE instance is a \NO instance.
 \end{claim}
 
 \begin{proof}
 Suppose not, then let us assume that $\RR = (R_i)_{i\in[n]}$ be a partial profile such that $\suc \in \CC(\RR)$ and $\sss(\RR)=\{c\}$. Let $s-MR(c,\RR) = s(s(\suc^\pr), \suc^\pr) - s(c, \suc^\pr)$ for some $\suc^\pr = (\suc_i^\pr)_{i\in[n]} \in \CC(\RR)$. We now claim the following.
 
 \begin{claim}\label{clm:beta}
  $s(c, \suc^\pr) = \beta_n$.
 \end{claim}
 
 \begin{proof}
 The idea of the proof is that if $s(c, \suc^\pr) > \beta_n$, then we can construct another profile which can be used to calculate worse regret for $c$ than $\suc^\pr$ and this will contradict the choice of $\suc^\pr$. Formally, let us define another profile $\bar{\suc} = (\bar{\suc_i})_{i\in[n]}$ where $\bar{\suc_i}$ is obtained from $\suc_i^\pr$ by ``moving'' $c$ immediately to the right of the alternatives that are on the left of $c$ in either $\suc_i^\pr$ or $\suc_i$ for $i\in[n]$; that is, for every $i\in[n]$, $\bar{\suc_i}$ is defined as follows.
 \begin{align*}
  \bar{\suc_i} &= \{(a,b): a,b\in\AA\setminus\{c\}, a\suc_i^\pr b\}\\
  &\cup\{(c,a): a\in\AA, c\suc_i^\pr a, c\suc_i a\}\\
  &\cup\{(a,c):a\in\AA, a\suc_i^\pr c \text{ or } a\suc_i c\}
 \end{align*}
 
 The profile $\bar{\suc}\in\CC(\RR)$ since $\suc\in\CC(\RR)$ and $\suc^\pr\in\CC(\RR)$. Due to monotonicity of $s$, the score of $c$ in $\bar{\suc}$ is at most the score of $c$ in $\suc^\pr$ and the score of every other alternative in $\bar{\suc}$ is at least their score in $\suc^\pr$. However, $\suc^\pr$ has been used to calculate the MR score of $c$ under $s$. Hence, we have the following:
 \[s(s(\suc^\pr), \suc^\pr) = s(s(\bar{\suc}), \suc^\pr),\quad s(c, \suc^\pr) = s(c, \bar{\suc})\]
 
 We now have the following:
 \begin{align*}
  \beta_n \le s(c, \suc^\pr) = s(c, \bar{\suc}) \le s(c,\suc)=\beta_n
 \end{align*}
 The first inequality follows from the definition of $\beta_n$ and the second inequality follows from monotonicity of $s$.
%
 \end{proof}
 
 Let $y\in s(\suc^\pr)$ and $s-MR(y,\RR) = s(s(\suc^\prr), \suc^\prr) - s(y, \suc^\prr)$ for some $\suc^\prr = (\suc_i^\prr)_{i\in[n]} \in \CC(\RR)$. We now have the following claim.
 
 \begin{claim}\label{clm:prr}
  $s(s(\suc^\prr), \suc^\prr) \le s(s(\suc^\pr), \suc^\pr)$
 \end{claim}

 \longversion{
 \begin{proof}
 The idea of the proof is along the same line as \Cref{clm:beta}. Let us define another profile $\bar{\suc} = (\bar{\suc_i})_{i\in[n]}$ where $\bar{\suc_i}$ is obtained from $\suc_i^\prr$ by ``moving'' $c$ immediately to the right of the alternatives that are on the left of $c$ in either $\suc_i^\prr$ or $\suc_i^\pr$ for $i\in[n]$; that is, for every $i\in[n]$, $\bar{\suc_i}$ is defined as follows.
 \begin{align*}
  \bar{\suc_i} &= \{(a,b): a,b\in\AA\setminus\{c\}, a\suc_i^\prr b\}\\
  &\cup\{(c,a): a\in\AA, c\suc_i^\prr a, c\suc_i^\pr a\}\\
  &\cup\{(a,c):a\in\AA, a\suc_i^\prr c \text{ or } a\suc_i^\pr c\}
 \end{align*}
 
 The profile $\bar{\suc}\in\CC(\RR)$ since $\suc^\pr\in\CC(\RR)$ and $\suc^\prr\in\CC(\RR)$. Due to monotonicity of $s$, the score of $c$ in $\bar{\suc}$ is at most the score of $c$ in $\suc^\prr$ and the score of every other alternative in $\bar{\suc}$ is at least their score in $\suc^\prr$. Again due to monotonicity of $s$, the score of $c$ in $\bar{\suc}$ is at most the score of $c$ in $\suc^\pr$. However, $\suc^\pr$ has been used to calculate the MR score of $c$ under $s$. Hence, we have $s(s(\suc^\prr), \suc^\prr) \le s(s(\suc^\pr), \suc^\pr)$; otherwise we will have $s-Regret(c, \suc^\pr)<s-Regret(c, \bar{\suc})$ which is a contradiction.
%
 \end{proof}
 }
 
 We now combine \Cref{clm:beta,clm:prr} as follows to prove the main claim.
 \begin{align*}
  s-MR(y,\RR) &= s(s(\suc^\prr), \suc^\prr) - s(y, \suc^\prr)\\
  &\le s(s(\suc^\pr), \suc^\pr) - s(y, \suc^\prr)\\
  &\le s(s(\suc^\pr), \suc^\pr) - \beta_n\\
  &= s(s(\suc^\pr), \suc^\pr) - s(c, \suc^\pr)\\
  &= s-MR(c,\RR)
 \end{align*}
 
 The second line follows from \Cref{clm:prr}, the third line follows from the definition of $\beta_n$, and the fourth line follows from \Cref{clm:beta}. Hence we have $s-MR(y,\RR)\le s-MR(c,\RR)$ which contradicts our assumption that $\sss(\RR)=\{c\}$.
 \end{proof}
 
 We now show that if $c$ does not receive the worst possible score with $n$ preferences over \AA under $s$ from the profile $\suc$, then the instance is a \YES instance. To see this, let us consider the partial profile $\PP = (P_i)_{i\in[n]}$ as $P_i=\{c\suc y: c\suc_i y\}$ for every $i\in[n]$. Let \RR be any profile in $\CC(\PP)$. Since, the alternative $c$ does not receive the worst possible score with $n$ preferences over \AA under $s$ from the profile $\suc$, $s(c, \RR) < \beta_n$. Hence, if $\alpha$ is the best possible score with $n$ preferences over \AA under $s$, we have $s-MR(c,\PP)<\alpha-\beta_n$. On the other hand, for any alternative $y\in\AA\setminus\{c\}$, let us consider the profile $\QQ_y = (Q_i)_{i\in[n]}$ where $Q_i = c \suc \cdots \suc y$ for every $i\in[n]$. Now due to monotonicity of $s$, we have $s(c, \QQ_y) = \alpha$ and $s(y, \QQ_y) = \beta_n$. Hence, we have $s-MR(y,\PP) = \alpha-\beta_n$ for every $y\in\AA\setminus\{c\}$ and thus $\sss(\PP) = \{c\}$.
\end{proof}

We remark that the proof of \Cref{thm:all_poly} for the co-winner case is trivial: every instance is a \YES instance since a partial profile where every preference is empty makes every alternative win due to neutrality. Since scoring rules, maximin, Copeland$^\alpha$ for every $\alpha\in[0,1]$, and simplified Bucklin voting rules are monotone, neutral, and worst efficient, \Cref{thm:all_poly} immediately implies the following corollary.

\begin{corollary}\label{cor:all_poly}
 The \SE problem is polynomial time solvable for scoring rules, maximin, Copeland$^\alpha$ for every $\alpha\in[0,1]$, and simplified Bucklin voting rules.
\end{corollary}

\longversion{
\begin{proof}
 The worst possible score with $n$ preferences over $m$ alternatives is $0$ for scoring rules, maximin, and Copeland $^\alpha$ for every $\alpha\in[0,1]$, and $-m$ for the simplified Bucklin voting rule (under modified but equivalent definition of simplified Bucklin voting rule).
\end{proof}
}

\section{Hardness Results for \SECPL}\label{sec:results_hard}

In this section we show that the \SECPL problem is \NPC for maximin, Copeland$^\alpha$ for every $\alpha\in[0,1]$, simplified Bucklin, and a large class of scoring rules which includes the $k$-approval voting rule for every $k$ and the Borda voting rule. \shortversion{In the interest of space, we omit some of our proofs. They are available in the supplementary material.}

\longversion{\subsection{Scoring Rules}}

Let us define a restricted version of the classical set cover problem which we call \SCFT. We will see in \Cref{lem:scft_hard} that this problem is \NPC by reducing from the vertex cover problem which is well known to be \NPC~\cite{garey1979computers}. Most of our \NP-hardness reductions are from this problem.
\begin{definition}[\SCFT]
 Given a universe \UU of cardinality $q$, a family $\SS=\{S_i:i\in[t]\}$ of $t$ subsets of \UU such that for every $a\in\UU$, we have $|\{i\in[t]:a\in S_i\}|=2$, and a positive integer \el, compute if there exists a subset $\GG\subseteq\SS$ containing at most \el sets such that $\cup_{A\in\GG} A = \UU$. We denote an arbitrary instance of \SCFT by $(\UU,\SS,\el)$.
\end{definition}

\begin{lemma}\label{lem:scft_hard}
 \SCFT is \NPC.
\end{lemma}
\longversion{
\begin{proof}
 \SCFT clearly belongs to \NP. To prove \NP-hardness, we reduce from \VC to \SCFT. Let $(\GG=(\VV,\EE),k)$ be an arbitrary instance of \VC. We construct and instance $(\UU,\SS,\el)$ as follows:
 \begin{align*}
  \UU = \{a_e: e\in\EE\},\SS= \{S_v: v\in\VV\}\text{ where } S_v=\{a_e: e \text{ is incident on } v\}, \el=k
 \end{align*}
 
 Clearly every element $a_e\in\UU$ belongs to exactly two sets, namely $S_u$ and $S_v$ if $e=\{u,v\}$. Also the equivalence of two instances are straight forward.
\end{proof}
}

We begin with showing that the \SECPL problem is \NPC for a large class of scoring rules which included the $k$-approval voting rule for every $3\le k\le \gamma m$ for any constant $0<\gamma<1$ and the Borda voting rule. While describing a (complete) preference, if we do not mention the order of any two alternatives, they can be ordered arbitrarily. On the other hand, if we are describing a partial preference and we do not mention the order of any two alternatives, then they should be assumed to be incomparable.

\begin{theorem}\label{thm:secpl_sc}
 Let $r$ be a normalized scoring rule such that there exists a function $g:\NB\longrightarrow\NB$ such that for every $m\in \NB$, we have $3m\le g(m)\le poly(m)$ and if $\alpha=(\alpha_i)_{i\in[g(m)]}$, then there exists a positive integer \ppp such that $3\le \ppp\le g(m)-m+3, \alpha_\ppp>\alpha_{\ppp+1}$ and $\alpha_{\ppp-1}=poly(m)$. Then the \SECPL problem is \NPC for the scoring rule $r$.
\end{theorem}

\longversion{
\begin{proof}
 The \SECPL problem clearly belongs to \NP. To prove \NP-hardness, we reduce from \SCFT to \SECPL for the scoring rule $r$.  Let $(\UU=\{u_1, \ldots, u_q\}, \SS=\{S_i: i\in[t]\},\el)$ be an arbitrary instance of \SCFT. Let us consider the following instance $(\AA, \PP, c, \lll)$ of \SECPL where \AA is defined as follows.
 \longversion{\[ \AA = \{a_i:i\in[q]\}\cup\{c,d\} \cup W, \text{ where } W=\{w_1, \ldots, w_{g(q)-q-2}\}\]}
 \shortversion{
 \begin{align*}
  \AA &= \{a_i:i\in[q]\}\cup\{c,d\} \cup W,\\
  &\text{ where } W=\{w_1, \ldots, w_{g(q)-q-2}\}
 \end{align*}
 }
 
 The profile \PP consists of the following preferences. For an integer $0\le k\le g(q)-q-2$, we denote the set $\{w_i:i\in[k]\}$ by $W_k$. Let $\kappa=\max\{i\in[g(q)]:\alpha_i\ne 0\}$; we observe that $\alpha_\kappa=1$ since $r$ is normalized. For $X\subseteq\UU$, let us denote the set $\{a_j: u_j\in X\}$ of alternatives by $X$ to simplify notation.
 \begin{itemize}
  \item $\forall i\in[t]: W_{\ppp-3} \suc S_i\suc d\suc c\suc (\UU\setminus S_i)\suc (W\setminus W_{\ppp-3})$
  \item $t-2$ copies of $W_{\ppp-1}\suc a_i\suc d\suc (\UU\setminus\{a_i\})\suc (W\setminus W_{\ppp-1})\suc c$
  \item $t(\alpha_\ppp-\alpha_{\ppp+2})$ copies of $W_2\suc \cdots\suc c\suc \cdots d$ where the alternative $c$ is placed at $\kappa$ position from left.
  \item If $\alpha_{g(q)-1}=1$, then we add $(q+1)t\alpha_{\ppp-1}$ copies of $\cdots\suc c\suc d$
  \item Otherwise:
  \begin{itemize}
   \item If $\kappa\le g(q)-q-1$, then we add $(q+1)t\alpha_{\ppp-1}$ copies of $W_2\suc\cdots\suc c\suc d\suc \UU\suc\cdots$ where the alternative $d$ is placed at $\kappa+1$ position from left and we add, for every $i\in[q]$, $(q+1)t\alpha_{\ppp-1}$ copies of $W_2\suc\cdots\suc a_i\suc d\suc c\suc (\UU\setminus\{a_i\})\suc\cdots$ where the alternative $d$ is placed at $\kappa+1$ position from the left.
   \item Otherwise we add $(q+1)t\alpha_{\ppp-1}$ copies of $W_2\suc\cdots\suc \UU\suc c\suc d\suc \cdots$ where the alternative $d$ is placed at $\kappa+1$ position from the left.
  \end{itemize}
 \end{itemize}

 For ease of reference, we call the above four groups as $\GG_1, \GG_2, \GG_3,$ and $\GG_4$ respectively. Let $n$ be the number of preferences in \PP. We observe that $n=poly(m)$ since $\alpha_{\ppp-1}=poly(m)$. The function \lll is defined as follows: $\lll(\{d,x\})=n$ for every $x\in\AA\setminus(\{c, d\}), \lll(\{d,c\})=n-\el$; the value of \lll be $0$ for all other pairs of alternatives. This finishes the description of our reduced instance. We now claim that the two instances are equivalent.
 
 In one direction, let us assume that the \SCFT instance is a \YES instance; without loss of generality, let us assume (by renaming) that $S_1, \ldots, S_\el$ forms a set cover of \UU. Let us consider the following partial profile \QQ with $\PP\in\LL(\QQ)$.
 
 \begin{itemize}
  \item Preferences in $\GG_1$: $\forall i\in[\el]: ((W_{\ppp-3} \cup S_i)\suc d\suc ((\UU\setminus S_i)\cup (W\setminus W_{\ppp-3})) \bigcup c\suc ((\UU\setminus S_i))\cup ((W\setminus W_{\ppp-3})))$
  
  \item Preferences in $\GG_1$: $\forall i \text{ with } \el+1\le i\le t: (W_{\ppp-3} \cup S_i)\suc d\suc c\suc ((\UU\setminus S_i)\cup (W\setminus W_{\ppp-3})$
  
  \item Preferences in $\GG_2$: $t-2$ copies of $(W_{\ppp-1}\cup a_i)\suc d\suc ((\UU\setminus\{a_i\})\cup (W\setminus W_{\ppp-1})\cup \{c\})$
  
  \item Preferences in $\GG_3$: $t(\alpha_\ppp-\alpha_{\ppp+2})$ copies of $c\suc X$ where the alternative $X=\{b\in\AA: c\suc b \text{ in } \GG_3\}$.
  
  \item Preferences in $\GG_4$: for every preference in $\GG_4$, we add $c\suc Y$ where the alternative $Y=\{b\in\AA: c\suc b \text{ in the corresponding preference in } \GG_4\}$.
 \end{itemize}
 
 Let $\Delta$ be the score that the alternative $w_1$ receives in \PP. We observe that the minimum scores that the alternatives $c$ and $a_i, i\in[q]$ receive in profile \RR with $\RR\in\LL(\QQ)$ are all the same; let it be $\lambda$. We summarize the MR score (based on $r$) of every alternative from \QQ in \Cref{tbl:scores_secpl_sc}. Hence the alternative $c$ wins uniquely in \QQ.
 \begin{table}[!htbp]
 \begin{center}
 \begin{tabular}{|c|c|c|c|}\hline
  
   \shortversion{Alternative & \makecell{MR-$r$ score\\ from \QQ} & \makecell{Competing\\ alternative}\\\hline\hline}
   \longversion{Alternative & MR-$r$ score from \QQ & Competing alternative\\\hline\hline}
   
   $c$ & $\Delta-t\alpha_\el -\lambda$ & $w_1$ (or $w_2$) \\\hline
   $a_i, \forall i\in[q]$ &$\Delta-t\alpha_\el+\alpha_{\el+1} -\lambda$ & $w_1$ (or $w_2$)\\\hline
   $w_1 (w_2)$ &$\Delta$& $w_2 (w_1)$ \\\hline
   $w\in W\setminus W_2$ & $\Delta$ & $w_1$\\\hline
   $d$ &$>D-t\alpha_\el$ & $w_1$ (or $w_2$)\\\hline
  \end{tabular}
  \caption{Summary of MR scores (based on $r$) of all the alternatives from the partial profile \QQ in the proof of \Cref{thm:secpl_sc}.}\label{tbl:scores_secpl_sc}
 \end{center}
 \end{table}
 
 In the other direction, let us assume that the \SECPL instance $(\AA, \PP, c, \lll)$ is a \YES instance. Let $\QQ$ be a partial profile such that $\PP\in\CC(\QQ)$ and the alternative $c$ wins uniquely in \QQ under the MR scoring rule based on $r$. We observe that if a preference profile $\RR_c$ with $\RR_c\in\LL(\QQ)$ is used to calculate the MR score of the alternative based on $r$, then the MR score of $c$ based on $r$ is at least $\Delta-t\alpha_\el -\lambda$ using the alternative $w_1$ as a competing alternative where $\lambda$ and $\Delta$ are as defined above. Let $J\subseteq[t]$ be the set of $i\in[t]$ such that the corresponding partial preferences in the group $\GG_1$ in \QQ leave the alternatives $c$ and $d$ incomparable. Since $\lll(\{d,c\})=n-\el$, we have $|J|\le\el$. We claim that $\{S_j: j\in J\}$ forms a set cover of \UU. Suppose not, then let $u_k\in\UU\setminus(\cup_{j\in J} S_j)$. Then we observe that the MR score of $a_k$ based on $r$ is at least $\Delta-t\alpha_\el -\lambda$ using the alternative $w_1$ as a competing alternative.However, this contradicts our assumption that $c$ is the unique MR-$r$ winner of \QQ. Hence $\{S_j: j\in J\}$ forms a set cover of \UU and thus the \SCFT is a \YES instance. This concludes the proof of the theorem.
\end{proof}
}

\Cref{thm:secpl_sc} immediately gives us the following corollary.

\begin{corollary}\label{cor:kapp_borda}
 The \SECPL problem is \NPC for the Borda and $k$-approval voting rules for every $3\le k\le \gamma m$ for any constant $0<\gamma<1$.
\end{corollary}

A drawback of \Cref{thm:secpl_sc} is that it does not cover the plurality, $2$-approval, and the $k$-veto voting rules for $k=o(m)$. We will show that the \SECPL problem is \NPC for the $k$-veto voting rule for any $1\le k\le \gamma m$ for any constant $0<\gamma<1$ in \Cref{thm:secpl_veto}. We now show in \Cref{thm:secpl_plurality_2app} that the \SECPL problem is \NPC for the plurality and $2$-approval voting rules by reducing it from the \XTC problem which is defined as follows and known to be \NPC~\cite{garey1979computers}.

\begin{definition}[\XTC]
 Given a universe \UU of cardinality $q$ such that $q$ is divisible by $3$, a family $\SS=\{S_i:i\in[t]\}$ of $t$ subsets of \UU each of cardinality $3$, compute if there exists a subset $\GG\subseteq\FF$ of $\nfrac{q}{3}$ sets such that $\cup_{A\in\GG} A = \UU$. We denote an arbitrary instance of \XTC by $(\UU,\SS)$.
\end{definition}

\begin{theorem}\label{thm:secpl_plurality_2app}
 The \SECPL problem is \NPC for the plurality and the $2$-approval voting rules.
\end{theorem}

\begin{proof}
 Let us first consider the plurality voting rule. The \SECPL problem for the plurality voting rule clearly belongs to \NP. To prove \NP-hardness, we reduce from \XTC to \SECPL for the plurality voting rule. Let $(\UU=\{u_1, \ldots, u_q\}, \SS=\{S_i: i\in[t]\})$ be an arbitrary instance of \XTC. For every $i\in[q]$ let us define $f_i=|\{j\in[t]:u_i\in S_j\}|$. Let us assume, without loss of generality, that $f_i<t-\nfrac{q}{2}$ (if not, then we add $3t$ new elements in \UU and $t$ sets in \SS each of size $3$ and collectively covering these new $3t$ elements). Let us assume, without loss of generality, that $q$ is divisible by $6$; if not then we add $3$ new elements in \UU and a set consisting of these three new elements in \SS. Let us consider the following instance $(\AA, \PP, c, \lll)$ of \SECPL where \AA is defined as follows.
 \[ \AA = \{a_i:i\in[q]\}\cup\{c, d, w\} \]
 
 The profile \PP consists of the following preferences. For $X\subseteq\UU$, let us also denote, for the sake of simplicity of notation, the set $\{a_j: u_j\in X\}$ of alternatives by $X$.
 \begin{itemize}
  \item $\forall i\in[t]: d\suc S_i \suc c\suc (\UU\setminus S_i)\suc w$
  \item $\nfrac{q}{6}+1$ copies of $c\suc (\AA\setminus\{c,w\})$
  \item $1$ copy of $d\suc w \suc c\suc (\AA\setminus\{d,w\})$
 \end{itemize}
 For ease of reference, we call the above three groups as $\GG_1,\GG_2,$ and $\GG_3$ respectively. Let $n$ be the number of preferences in \PP. That is, $n=t+\nfrac{q}{6}+2$. The function \lll is defined as follows: $\lll(\{c,a_i\})=n-f_i+1$ for every $i\in[q]$, $\lll(\{w,d\})=n-1,$ $\lll(\{w,x\})=n$ for every $x\in\AA\setminus\{w,d\}$; the value of \lll be $0$ for all other pairs of alternatives. This finishes the description of our reduced instance. We now claim that the two instances are equivalent.
 
 In one direction, let us assume that the \XTC instance is a \YES instance; without loss of generality, let us assume (by renaming) that $S_1, \ldots, S_{\nfrac{q}{3}}$ forms a set cover of \UU. Let us consider the following partial profile \QQ with $\PP\in\LL(\QQ)$.
 
 \begin{itemize}
  \item Preferences in $\GG_1$: $\forall i\in[\nfrac{q}{3}]: (S_i \suc c \suc (\UU\setminus S_i)\suc w) \bigcup (d\suc w)$
  
  \item Preferences in $\GG_1$: $\forall i \text{ with } \nfrac{q}{3}+1\le i\le t: (c \suc (\UU\setminus S_i)\suc w)\bigcup (d\suc w)$
  
  \item Preferences in $\GG_2$: $\nfrac{q}{6}+1$ copies of $c\suc (\AA\setminus\{c,w\})$
  
  \item Preferences in $\GG_3$: $1$ copy of $w \suc c\suc (\AA\setminus\{d,w\})$
 \end{itemize}
 
 We summarize the MR-plurality score of every alternative from \QQ in \Cref{tbl:scores_secpl_plurality}. Hence the alternative $c$ wins uniquely in \QQ.
 \begin{table}[!htbp]
 \begin{center}
 \begin{tabular}{|c|c|c|c|}\hline
  
  \shortversion{Alternative & \makecell{MR-plurality score\\ from \QQ} & \makecell{Competing\\ alternative}\\\hline\hline}
   \longversion{Alternative & MR-plurality score from \QQ & Competing alternative\\\hline\hline}
   
   $c$ & $t-\nfrac{q}{6}$ & $d$ \\\hline
   $a_i, \forall i\in[q]$ &$t$ & $d$\\\hline
   $d$ & $t-\nfrac{q}{6}+1$ & $c$\\\hline
  \end{tabular}
  \caption{Summary of MR-plurality scores of all the alternatives from the partial profile \QQ in the proof of \Cref{thm:secpl_plurality_2app}.}\label{tbl:scores_secpl_plurality}
 \end{center}
 \end{table}
 
 In the other direction, let us assume that the \SECPL instance $(\AA, \PP, c, \lll)$ is a \YES instance. Let $\QQ$ be a partial profile such that $\PP\in\CC(\QQ)$ and the alternative $c$ wins uniquely in \QQ under the MR-plurality voting rule. Let $J\subseteq[t]$ be the set of $i\in[t]$ such that the corresponding partial preferences in the group $\GG_1$ in \QQ leave the alternatives $c$ and at least one alternative in $S_i$ incomparable. A key observation is that since $\lll(\{c,a_i\})=n-f_i+1$ for $i\in[q]$, we have $\cup_{j\in J} S_j = \UU$. Hence we have $|J|\ge \nfrac{q}{3}$. We now claim that $|J|\le \nfrac{q}{3}$. Suppose not, then the MR-plurality score of $d$ is at most $(t-\nfrac{q}{3}-1)+\nfrac{q}{6}+1=t-\nfrac{q}{6}$ using $c$ as competing alternative (we observe that, since $\lll(\{c,a_i\})=n-f_i+1$ for $i\in[q]$, using the alternative $a_i$ as a competing alternative for any $i$ will lead to MR-plurality score of $d$ at most $2f_i<t-\nfrac{q}{6}$). Hence the MR-plurality score of $d$ is at most $t-\nfrac{q}{6}$. However the MR-plurality score of $c$ is at least $t+1-(\nfrac{q}{6}+1)=t-\nfrac{q}{6}$. This contradicts our assumption that $c$ is the unique MR-plurality winner of \QQ. Hence $\{S_j: j\in J\}$ forms a set cover of \UU and thus the \XTC is a \YES instance. This concludes the proof of the theorem.
 
 For the $2$-approval voting rule, we can introduce $n$ dummy alternatives each of which appears at the first position in exactly one preference and in the rest $(n-1)$ preferences, it appears in the bottom $(n-1)$ positions. All other parameters of the reduction remain same. It is easy to see that a similar argument will prove the result for the $2$-approval voting rule.
\end{proof}

We now show our hardness result for the $k$-veto voting rule by reducing from \XTC.
\begin{theorem}\label{thm:secpl_veto}
 The \SECPL problem is \NPC for the $k$-veto voting rule for every $1\le k\le \gamma m$ for any constant $0<\gamma<1$.
\end{theorem}

\begin{proof}
 Let us first consider the veto voting rule. The \SECPL problem for the veto voting rule clearly belongs to \NP. To prove \NP-hardness, we reduce from \XTC to \SECPL for the veto voting rule. Let $(\UU=\{u_1, \ldots, u_q\}, \SS=\{S_i: i\in[t]\})$ be an arbitrary instance of \XTC. For every $i\in[q]$ let us define $f_i=|\{j\in[t]:u_i\in S_j\}|$. Let us assume, without loss of generality, that $f_i<t-\nfrac{q}{2}$ (if not, then we add $3t$ new elements in \UU and $t$ sets in \SS each of size $3$ and collectively covering these new $3t$ elements). Let us assume, without loss of generality, that $q$ is divisible by $6$; if not then we add $3$ new elements in \UU and a set consisting of these three new elements in \SS. We also assume that $t$ is an odd integer; if not then we duplicate one set in \SS. Let us consider the following instance $(\AA, \PP, c, \lll)$ of \SECPL where \AA is defined as follows.
 \[ \AA = \{a_i:i\in[q]\}\cup\{c, x, d, w_1, w_2\} \]
 
 The profile \PP consists of the following preferences. For $X\subseteq\UU$, let us also denote, for the sake of simplicity of notation, the set $\{a_j: u_j\in X\}$ of alternatives by $X$.
 \begin{enumerate}
  \item $\forall i\in[t]: (\UU\setminus S_i)\suc w_1\suc w_2\suc d\suc c\suc S_i\suc x$
  \item $\nfrac{(t-1)}{2}+\nfrac{q}{6}$ copies of $(\AA\setminus\{c,d\})\suc d\suc c$
  \item $\forall j\in[q]: \nfrac{(t+1)}{2}+\nfrac{q}{6}+1-f_j$ copies of $(\AA\setminus\{a_j,d\})\suc d\suc a_j$
  \item $10t$ copies of $(\AA\setminus\{d\})\suc d$
  \item $10t$ copies of $(\AA\setminus\{d,w_1,w_2\})\suc d\suc w_1\suc w_2$
  \item $\forall i\in[2]:\nfrac{q}{3}$ copies of $(\AA\setminus\{w_i,d\})\suc d\suc w_i$
 \end{enumerate}
 For ease of reference, we call the above six groups as $\GG_1,\GG_2,\GG_3,\GG_4,\GG_5$ and $\GG_6$ respectively. Let $n$ be the number of preferences in \PP. The function \lll is defined as follows: $\lll(\{d,z\})=n$ for every $z\in\AA\setminus\{d\}$; the value of \lll be $0$ for all other pairs of alternatives. This finishes the description of our reduced instance. We now claim that the two instances are equivalent.
 
 In one direction, let us assume that the \SCFT instance is a \YES instance; without loss of generality, let us assume (by renaming) that $S_1, \ldots, S_{\nfrac{q}{3}}$ forms a set cover of \UU. Let us consider the following partial profile \QQ with $\PP\in\LL(\QQ)$.
 
 \begin{itemize}
  \item Preferences in $\GG_1$: $\forall i\in[\nfrac{q}{3}]: (\UU\setminus S_i)\suc w_1\suc w_2\suc d\suc c\suc S_i\suc x$
  
  \item Preferences in $\GG_1$: $\forall i \text{ with } \nfrac{q}{3}+1\le i\le t: (\UU\setminus S_i)\suc w_1\suc w_2\suc d\suc c\suc (S_i\cup \{x\})$
  
  \item Preferences in $\GG_2$: $\nfrac{(t-1)}{2}+\nfrac{q}{6}$ copies of $(\AA\setminus\{c,d\})\suc d\suc c$
  
  \item Preferences in $\GG_3$: $\forall j\in[q]: \nfrac{(t+1)}{2}+\nfrac{q}{6}+1-f_j$ copies of $(\AA\setminus\{a_j,d\})\suc d\suc a_j$
  
  \item Preferences in $\GG_4$: $10t$ copies of $(\AA\setminus\{d,w_1,w_2\})\suc d\suc \{w_1, w_2\}$
  
  \item Preferences in $\GG_5$: $10t$ copies of $(\AA\setminus\{d\})\suc d$
  
  \item Preferences in $\GG_6$: $\forall i\in[2]:\nfrac{q}{3}$ copies of $(\AA\setminus\{w_i,d\})\suc d\suc w_i$
 \end{itemize}
 
 We summarize the MR-veto score of every alternative from \QQ in \Cref{tbl:scores_secpl_veto}. Hence the alternative $c$ wins uniquely in \QQ.
 \begin{table}[!htbp]
 \begin{center}
 \begin{tabular}{|c|c|c|c|}\hline
  
   \shortversion{Alternative & \makecell{MR-veto score\\ from \QQ} & \makecell{Competing\\ alternative}\\\hline\hline}
   \longversion{Alternative & MR-veto score from \QQ & Competing alternative\\\hline\hline}
  
   
   $c$ & $\nfrac{(t-1)}{2}-\nfrac{q}{6}$ & $x$ (or $w_1$ or $w_2$) \\\hline
   $x$ & $\nfrac{(t+1)}{2}-\nfrac{q}{6}$ & $c$ \\\hline
   $a_i, \forall i\in[q]$ & $\nfrac{(t+1)}{2}-\nfrac{q}{6}$ & $x$ (or $w_1$ or $w_2$)\\\hline
   $d$ & $\ge 5t$ & $x$ (or $w_1$ or $w_2$)\\\hline
   $w_1 (w_2)$ & $\ge 5t$& $w_2(w_1)$\\\hline
  \end{tabular}
  \caption{Summary of MR-veto scores of all the alternatives from the partial profile \QQ in the proof of \Cref{thm:secpl_veto}.}\label{tbl:scores_secpl_veto}
 \end{center}
 \end{table}
 
 In the other direction, let us assume that the \SECPL instance $(\AA, \PP, c, \lll)$ is a \YES instance. Let $\QQ$ be a partial profile such that $\PP\in\CC(\QQ)$ and the alternative $c$ wins uniquely in \QQ under the MR-veto voting rule. We first observe that, since $\lll(\{d,z\})=n$ for every $z\in\AA\setminus\{d\}$, the alternative $c$ is forced to receive a score of $0$ from every preference in \QQ corresponding to group $\GG_2$ (and thus every other alternative is forced to receive a score of $1$), both the alternatives $x$ and $c$ are forced to receive a score of $1$ from every preference in \QQ corresponding to group $\GG_3$, and the alternative $d$ is forced to receive a score of $0$ from every preference in \QQ corresponding to group $\GG_4$ (and thus every other alternative is forced to receive a score of $1$). We now claim that there must be at least $\nfrac{q}{3}$ preferences in \QQ corresponding to the group $\GG_1$ where the alternative $x$ must be forced to receive a score of $0$ for $c$ to win uniquely. Suppose not, then we observe that the MR-veto score of $c$ is at least $\nfrac{(t+1)}{2}-\nfrac{q}{6}$ using $x$ as a competing alternative whereas the MR-veto score of $x$ is at most $\nfrac{(t-1)}{2}-\nfrac{q}{6}$ using $c$ as a competing alternative which contradicts our assumption that $c$ wins the MR-veto election. Let $J\subseteq[t]$ be the set of $i\in[t]$ such that the corresponding partial preferences in the group $\GG_1$ in \QQ force the alternative $x$ to receive a score of $0$. Then we have $|J|\ge\nfrac{q}{3}$. We now claim that $|J|=\nfrac{q}{3}$ and the sets $\{S_j:j\in J\}$ forms an exact set cover of \UU. Suppose not, then there exists an $u_k\in\UU$ which appears in at least two sets in $\{S_j:j\in J\}$. Then the alternative is forced to receive a score of $1$ from at least $t-f_k+2$ preferences in \QQ corresponding to the group $\GG_1$ -- $t-f_k$ preferences where it appears on the left of $d$ and at least $2$ preferences where the alternative $x$ is forced to receive a score of $0$. However, then the MR-veto score of the alternative $a_k$ is at most $\nfrac{(t-1)}{2}-\nfrac{q}{6}$ using $w_1$ (or $w_2$) as a competing alternative. This contradicts our assumption that $c$ wins the MR-veto election uniquely. Hence the sets $\{S_j:j\in J\}$ forms an exact set cover of \UU and thus the \XTC is a \YES instance. This concludes the proof of the theorem for the veto voting rule.
 
 For the $k$-veto voting rule, we add $k-1$ ``dummy'' alternatives $d_i, i\in[k-1]$ at the bottom of every preferences in \PP and this is easily verifiable that the proof for the $k$-veto voting rule goes along the same line.
\end{proof}

\longversion{\subsection{Maximin Voting Rule}}

We now show our hardness result for the maximin voting rule by reducing from \SCFT.
\begin{theorem}\label{thm:secpl_maximin}
 The \SECPL problem is \NPC for the maximin voting rule.
\end{theorem}

\begin{proof}
 The \SECPL problem clearly belongs to \NP. To prove \NP-hardness, we reduce from \SCFT to \SECPL for the maximin voting rule.  Let $(\UU=\{u_1, \ldots, u_q\}, \SS=\{S_i: i\in[t]\},\el)$ be an arbitrary instance of \SCFT. Let us consider the following instance $(\AA, \PP, c, \lll)$ of \SECPL where \AA is defined as follows.
 \[ \AA = \{a_i:i\in[q]\}\cup\{c, w_1, w_2, d\} \]
 
 The profile \PP consists of the following preferences. For $X\subseteq\UU$, let us also denote, for the sake of simplicity of notation, the set $\{a_j: u_j\in X\}$ of alternatives by $X$.
 \begin{itemize}
  \item $\forall i\in[t]: w_1\suc w_2\suc S_i \suc d\suc c\suc (\UU\setminus S_i)$
  \item $2$ copies of $c\suc \UU\suc d\suc w_1\suc w_2$
 \end{itemize}
 For ease of reference, we call the above two groups as $\GG_1$ and $\GG_2$ respectively. Let $n$ be the number of preferences in \PP. The function \lll is defined as follows: $\lll(\{d,x\})=n$ for every $x\in\AA\setminus(\{c, d\}), \lll(\{d,c\})=n-\el$; the value of \lll be $0$ for all other pairs of alternatives. This finishes the description of our reduced instance. We now claim that the two instances are equivalent.
 
 In one direction, let us assume that the \SCFT instance is a \YES instance; without loss of generality, let us assume (by renaming) that $S_1, \ldots, S_\el$ forms a set cover of \UU. Let us consider the following partial profile \QQ with $\PP\in\LL(\QQ)$.
 
 \begin{itemize}
  \item Preferences in $\GG_1$: $\forall i\in[\el]: ((\{w_1,w_2\}\cup S_i)\suc d\suc \{\UU\setminus S_i\})\bigcup (c\suc \{\UU\setminus S_i\}) $
  
  \item Preferences in $\GG_1$: $\forall i \text{ with } \el+1\le i\le t: (\{w_1,w_2\}\cup S_i)\suc d\suc c\suc \{\UU\setminus S_i\}$
  
  \item Preferences in $\GG_2$: $2$ copies of $c\suc \UU \suc d\suc \{w_1,w_2\}$
 \end{itemize}
 
 We summarize the MR-maximin score of every alternative from \QQ in \Cref{tbl:scores_secpl_maximin}. Hence the alternative $c$ wins uniquely in \QQ.
 
 \longversion{
 \begin{table}[!htbp]
 \begin{center}
 \begin{tabular}{|c|c|c|c|}\hline
  
   Alternative & MR-maximin from \QQ & Competing alternative & Comments\\\hline\hline
   
   $c$ & $t-2$ & $w_1$ (or $w_2$)& $N(w_1,c)-N(c,d)$ \\\hline
   $a_i, \forall i\in[q]$ &$t-1$ & $w_1$ (or $w_2$)& $N(w_1,c)-N(a_i,c)$\\\hline
   $w_1 (w_2)$ &$t$& $w_2 (w_1)$& $N(w_1,c)-N(w_2,w_1)$ \\\hline
   $d$ &$t$ & $w_1$ (or $w_2$)& $N(w_1,c)-N(d,w_1)$\\\hline
  \end{tabular}
  \caption{Summary of MR-maximin scores of all the alternatives from the partial profile \QQ in the proof of \Cref{thm:secpl_maximin}.}\label{tbl:scores_secpl_maximin}
 \end{center}
 \end{table}
 }
 \shortversion{
 \begin{table}[!htbp]
 \begin{center}
 \begin{tabular}{|c|c|c|c|}\hline
  
   Alternative & \makecell{MR-maximin\\ score from \QQ} & Comments\\\hline\hline
   
   $c$ & $t-2$ & $N(w_1,c)-N(c,d)$ \\\hline
   $a_i, \forall i\in[q]$ &$t-1$ & $N(w_1,c)-N(a_i,c)$\\\hline
   $w_1 (w_2)$ &$t$&  $N(w_1,c)-N(w_2,w_1)$ \\\hline
   $d$ &$t$ &  $N(w_1,c)-N(d,w_1)$\\\hline
  \end{tabular}
  \caption{Summary of MR-maximin scores of all the alternatives from the partial profile \QQ in the proof of \Cref{thm:secpl_maximin}.}\label{tbl:scores_secpl_maximin}
 \end{center}
 \end{table}
 }
 
 In the other direction, let us assume that the \SECPL instance $(\AA, \PP, c, \lll)$ is a \YES instance. Let $\QQ$ be a partial profile such that $\PP\in\CC(\QQ)$ and the alternative $c$ wins uniquely in \QQ under the MR-maximin voting rule. We observe that for every $\RR\in\LL(\QQ)$ which can be used for calculating the MR-maximin score of the alternative $c$, we have $N_\RR(c,d)\le 2$. Also, there are only two preferences (the preferences in $\GG_2$) where there exist some alternatives which are preferred over the alternative $w_1$. Hence the MR-maximin score of the alternative $c$ in \QQ is at least $t-2$. Let $J\subseteq[t]$ be the set of $i\in[t]$ such that the corresponding partial preferences in the group $\GG_1$ in \QQ leave the alternatives $c$ and $d$ incomparable. Since $\lll(\{d,c\})=n-\el$, we have $|J|\le\el$. We claim that $\{S_j: j\in J\}$ forms a set cover of \UU. Suppose not, then let $u_k\in\UU\setminus(\cup_{j\in J} S_j)$. We observe that for every $\RR^\pr\in\LL(\QQ)$, we have $N_{\RR^\pr}(a_k,c)=2$. We also observe that $N_{\RR^\pr}(a_k,d)=N_{\RR^\pr}(a_k,w_i)=2$ for every $i\in[2]$ since $\lll(\{d,u_i\})=n$ for every $i\in[q]$ and $\lll(\{d,w_1\}) = \lll(\{d,w_2\})=n$. Hence, the MR-maximin score of the alternative $a_k$ is $t-2$ where the alternative $w_1$ plays the role of a competing alternative. However, this contradicts our assumption that $c$ is the unique MR-maximin winner of \QQ. Hence $\{S_j: j\in J\}$ forms a set cover of \UU and thus the \SCFT is a \YES instance. This concludes the proof of the theorem.
\end{proof}

\longversion{\subsection{Copeland$^\alpha$ Voting Rule}}

\longversion{We show next our hardness result for the Copeland$^\alpha$ voting rule for every $\alpha\in[0,1]$.}
\shortversion{For the Copeland$^\alpha$ voting rule, we have the following result for every $\alpha\in[0,1]$.}
\begin{theorem}\label{thm:secpl_copeland}
 The \SECPL problem is \NPC for the Copeland$^\alpha$ voting rule for every $\alpha\in[0,1]$.
\end{theorem}

\longversion{
\begin{proof}
 The \SECPL problem clearly belongs to \NP. To prove \NP-hardness, we reduce from \SCFT to \SECPL for the Copeland$^\alpha$ voting rule for every $\alpha\in[0,1]$.  Let $(\UU=\{u_1, \ldots, u_q\}, \SS=\{S_i: i\in[t]\},\el)$ be an arbitrary instance of \SCFT. Let us consider the following instance $(\AA, \PP, c, \lll)$ of \SECPL where \AA is defined as follows.
 \[ \AA = \{a_i:i\in[q]\}\cup\{c, x, y, d, w_1, w_2\}\} \]
 
 The profile \PP consists of the following preferences. For $X\subseteq\UU$, let us also denote, for the sake of simplicity of notation, the set $\{a_j: u_j\in X\}$ of alternatives by $X$.
 \begin{itemize}
  \item $\forall i\in[t]: w_1\suc w_2\suc S_i \suc d\suc c\suc x\suc y\suc (\UU\setminus S_i)$
  \item $t-3$ copies of $\UU\suc c\suc x\suc y\suc d\suc w_1\suc w_2$
 \end{itemize}
 For ease of reference, we call the above two groups as $\GG_1$ and $\GG_2$ respectively. Let $n$ be the number of preferences in \PP. We observe that $n$ is an odd integer and thus the value of $\alpha$ is irrelevant. Hence, from here on, we omit the parameter $\alpha$. The function \lll is defined as follows: $\lll(\{d,x\})=n$ for every $x\in\AA\setminus\{c,d,x,w_1,w_2\}, \lll(\{d,c\})=\lll(\{d,x\})=n-\el$; the value of \lll be $0$ for all other pairs of alternatives. This finishes the description of our reduced instance. We now claim that the two instances are equivalent.
 
 In one direction, let us assume that the \SCFT instance is a \YES instance; without loss of generality, let us assume (by renaming) that $S_1, \ldots, S_\el$ forms a set cover of \UU. Let us consider the following partial profile \QQ.
 
 \begin{itemize}
  \item Preferences in $\GG_1$: $\forall i\in[\el]: S_i\suc d\suc y\suc (\UU\setminus S_i)) \bigcup (c\suc \{\{x,y\}\cup (\UU\setminus S_i)\})$
  
  \item Preferences in $\GG_1$: $\forall i \text{ with } \el+1\le i\le t: S_i\suc d\suc x\suc y\suc (\UU\setminus S_i)$
  \item Preferences in $\GG_2$: $t-3$ copies of $c\suc \{x, d, w_1, w_2\}$
 \end{itemize}
 
 We summarize the MR-Copeland score of every alternative from \QQ in \Cref{tbl:scores_secpl_copeland}. Hence the alternative $c$ wins uniquely in \QQ.
 \begin{table*}[!htbp]
 \begin{center}
 \begin{tabular}{|c|c|c|c|}\hline
  
   Alternative & MR-Copeland from \QQ & Competing alternative & Comments\\\hline\hline
   
   $c$ & $|\AA|-3$ & $w_1$ (or $w_2$)& \makecell{$w_1$ defeats $\AA\setminus\{w_1\}$\\$c$ defeats $\{x,y\}$} \\\hline
   $a_i, \forall i\in[q]$ &$|\AA|-2$ & $w_1$ (or $w_2$)& \makecell{$w_1$ defeats $\AA\setminus\{w_1\}$\\$a_i$ defeats $d$}\\\hline
   $w_1 (w_2)$ &$|\AA|-1$& $w_2 (w_1)$& \makecell{$w_1$ defeats $\AA\setminus\{w_1\}$\\$w_2$ defeats none} \\\hline
   $d$ &$|\AA|-2$& $w_2 (w_1)$& \makecell{$w_1$ defeats $\AA\setminus\{w_1\}$\\$d$ defeats $y$}\\\hline
  \end{tabular}
  \caption{Summary of MR-Copeland scores of all the alternatives from the partial profile \QQ in the proof of \Cref{thm:secpl_copeland}.}\label{tbl:scores_secpl_copeland}
 \end{center}
 \end{table*}
 
 In the other direction, let us assume that the \SECPL instance $(\AA, \PP, c, \lll)$ is a \YES instance. Let $\QQ$ be a partial profile such that $\PP\in\CC(\QQ)$ and the alternative $c$ wins uniquely in \QQ under the MR-Copeland voting rule. We observe that in \PP, the alternative $c$ defeats only $2$ alternatives namely $x$ and $y$. Since, $\PP\in\LL(\QQ)$, the MR-Copeland score of the alternative $c$ in \QQ is at least $(|\AA|-1)-2=|\AA|-3$. Let $J\subseteq[t]$ be the set of $i\in[t]$ such that the corresponding partial preferences in \QQ leave the alternatives $c$ and $d$ incomparable. Since $\lll(\{d,c\})=n-\el$ for every $i\in[t]$, we have $|J|\le\el$. We claim that $\{S_j: j\in J\}$ forms a set cover of \UU. Suppose not, then let $u_k\in\UU\setminus(\cup_{j\in J} S_j)$. We observe that in every profile $\RR\in\LL(\QQ)$, the alternative $a_k$ defeats the alternatives $c$ and $d$ and thus the MR-Copeland score of the alternative $a_k$ in \QQ is at most $(|\AA|-3)$ where the alternative $w_1$ plays the role of a competing alternative.  However, this contradicts our assumption that $c$ is the unique MR-Copeland winner of \QQ. Hence $\{S_j: j\in J\}$ forms a set cover of \UU and thus the \SCFT is a \YES instance. This concludes the proof of the theorem.
\end{proof}
}

\longversion{\subsection{Simplified Bucklin Voting Rule}}

\longversion{We now show our hardness result for the simplified Bucklin voting rule.}
\shortversion{For the simplified Bucklin voting rule, we have the following result.}

\begin{theorem}\label{thm:secpl_bucklin}
 The \SECPL problem is \NPC for the simplified Bucklin voting rule.
\end{theorem}

\longversion{
\begin{proof}
 The \SECPL problem clearly belongs to \NP. To prove \NP-hardness, we reduce from \SCFT to \SECPL for the simplified Bucklin voting rule.  Let $(\UU=\{u_1, \ldots, u_q\}, \SS=\{S_i: i\in[t]\},\el)$ be an arbitrary instance of \SCFT. Let us consider the following instance $(\AA, \PP, c, \lll)$ of \SECPL where \AA is defined as follows.
 \[ \AA = \{a_i:i\in[q]\}\cup\{c, d\} \cup W, \text{ where } W=\{w_i: i\in[q]\} \]
 
 The profile \PP consists of the following preferences. For $X\subseteq\UU$, let us also denote, for the sake of simplicity of notation, the set $\{a_j: u_j\in X\}$ of alternatives by $X$.
 \begin{itemize}
  \item $\forall i\in[t]: w_1\suc w_2 \suc \cdots \suc w_{q-2}\suc S_i \suc d\suc c\suc (\UU\setminus S_i)\suc w_{q-1}\suc w_q$
  \item $t-1$ copies of $\UU\suc c\suc d\suc W$
  \item $2$ copies of $w_1\suc w_2 \suc \cdots \suc w_q\suc c\suc d\suc\UU$
 \end{itemize}
 For ease of reference, we call the above three groups as $\GG_1,\GG_2,$ and $\GG_3$ respectively. Let $n$ be the number of preferences in \PP. The function \lll is defined as follows: $\lll(\{d,x\})=n$ for every $x\in\AA\setminus(\{c, d, w_1, \ldots, w_{q-2}\}), \lll(\{d,c\})=n-\el$; the value of \lll be $0$ for all other pairs of alternatives. This finishes the description of our reduced instance. We now claim that the two instances are equivalent.
 
 In one direction, let us assume that the \SCFT instance is a \YES instance; without loss of generality, let us assume (by renaming) that $S_1, \ldots, S_\el$ forms a set cover of \UU. Let us consider the following partial profile \QQ with $\PP\in\LL(\QQ)$.
 
 \begin{itemize}
  \item Preferences in $\GG_1$: $\forall i\in[\el]: (S_i \suc d\suc (\UU\setminus S_i)\suc w_{q-1}\suc w_q)\bigcup\left(c\suc (\UU\setminus S_i)\suc w_{q-1}\suc w_q\right)$
  
  \item Preferences in $\GG_1$: $\forall i \text{ with } \el+1\le i\le t: S_i \suc d\suc c\suc (\UU\setminus S_i)\suc w_{q-1}\suc w_q$
  
  \item Preferences in $\GG_2$: $t-1$ copies of $\UU\suc c\suc d\suc W$
  
  \item Preferences in $\GG_3$: $2$ copies of $w_{q-1}\suc w_q \suc c\suc d \suc \UU$
 \end{itemize}
 
 We summarize the MR-simplified Bucklin score of every alternative from \QQ in \Cref{tbl:scores_secpl_bucklin}. Hence the alternative $c$ wins uniquely in \QQ.
 \begin{table}[!htbp]
 \begin{center}
 \begin{tabular}{|c|c|c|c|}\hline
  
   \shortversion{Alternative & \makecell{MR-simplified Bucklin\\ score from \QQ} & \makecell{Competing\\ alternative}\\\hline\hline}
   \longversion{Alternative & MR-simplified Bucklin score from \QQ & Competing alternative\\\hline\hline}
   
   $c$ & $-q$ & $w_1$ \\\hline
   $a_i, \forall i\in[q]$ &$-q-1$ & $w_1$\\\hline
   $w_1$ &$-2q-1$& $w_2$ \\\hline
   $w_i, 2\le i\le q$ &$-2q-1$& $w_1$ \\\hline
   $d$ &$-q-1$ & $w_1$\\\hline
  \end{tabular}
  \caption{Summary of MR-simplified Bucklin scores of all the alternatives from the partial profile \QQ in the proof of \Cref{thm:secpl_bucklin}.}\label{tbl:scores_secpl_bucklin}
 \end{center}
 \end{table}
 
 In the other direction, let us assume that the \SECPL instance $(\AA, \PP, c, \lll)$ is a \YES instance. Let $\QQ$ be a partial profile such that $\PP\in\CC(\QQ)$ and the alternative $c$ wins uniquely in \QQ under the MR-simplified Bucklin voting rule. We can assume without loss of generality that the alternative $w_1$ is used as a competing alternative to calculate MR-simplified Bucklin score of every alternative (other than $w_1$ itself) since the simplified Bucklin score of $w_1$ is $1$ (the minimum possible) in \PP and $\PP\in\CC(\QQ)$. Let $J\subseteq[t]$ be the set of $i\in[t]$ such that the corresponding partial preferences in the group $\GG_1$ in \QQ leave the alternatives $c$ and $d$ incomparable. Since $\lll(\{d,c\})=n-\el$, we have $|J|\le\el$. We claim that $\{S_j: j\in J\}$ forms a set cover of \UU. Suppose not, then let $u_k\in\UU\setminus(\cup_{j\in J} S_j)$. We observe that for every $\RR\in\LL(\QQ)$, the simplified Bucklin score of the alternative $a_k$ is at most $q+1$. Hence the MR-simplified Bucklin score of $a_k$ in \QQ is at least $-q$. On the other hand, since the simplified Bucklin score of $c$ in \PP is $q+1$ and $\PP\in\CC(\QQ)$, the MR-simplified Bucklin score of $c$ in \QQ is at most $-q$. However, this contradicts our assumption that $c$ is the unique MR-simplified Bucklin winner of \QQ. Hence $\{S_j: j\in J\}$ forms a set cover of \UU and thus the \SCFT is a \YES instance. This concludes the proof of the theorem.
\end{proof}
}

%% file: conclusion.tex
\section{Conclusion and Future Work}

In this work, we have discovered an important vulnerability, namely manipulative elicitation, in the use of minimax regret based extension of classical voting rules in the incomplete preferential setting. Moreover, we have shown that the related computational task is polynomial time solvable for many commonly used voting rules\longversion{ including all scoring rules, maximin, Copeland$^\alpha$ for every $\alpha\in[0,1]$, simplified Bucklin voting rules, etc}. Then we have shown that by introducing a parameter per pair of alternatives which specifies the minimum number of partial preferences where this pair of alternatives must be comparable makes the computational task of manipulative elicitation \NPC for \longversion{all the above mentioned}\shortversion{many common} voting rules.\longversion{ We want to draw special attention to the fact that our approach makes manipulative elicitation \NPC even for the plurality and veto voting rules which are vulnerable to most of the other manipulative attacks. In summary, we have found an important vulnerability in the incomplete preferential setting and proposed a novel approach to tackle it.}

A drawback of our approach is that the parameters can be non-uniform -- their values do not need to be the same for every pair of alternatives. It would be interesting to study the computational complexity of the problem when the values of the parameters are all the same. In another direction, it would be interesting to conduct extensive experimentation to study usefulness of our approach in practice. This is specially important since computational intractability is known to provide only a weak barrier in other forms of election manipulation~\cite{procaccia2007junta}.

%% file: elicitation.bbl
\newcommand{\etalchar}[1]{$^{#1}$}
\begin{thebibliography}{PRVW07}

\bibitem[BBN10]{betzler2010partial}
Nadja Betzler, Robert Bredereck, and Rolf Niedermeier.
\newblock Partial kernelization for rank aggregation: theory and experiments.
\newblock In {\em Proc. International Symposium on Parameterized and Exact
  Computation (IPEC)}, pages 26--37. Springer, 2010.

\bibitem[BBN14]{BetzlerBN14}
Nadja Betzler, Robert Bredereck, and Rolf Niedermeier.
\newblock Theoretical and empirical evaluation of data reduction for exact
  kemeny rank aggregation.
\newblock {\em Auton. Agent Multi Agent Syst.}, 28(5):721--748, 2014.

\bibitem[BHN09]{betzler2009multivariate}
Nadja Betzler, Susanne Hemmann, and Rolf Niedermeier.
\newblock A {M}ultivariate {C}omplexity {A}nalysis of {D}etermining {P}ossible
  {W}inners given {I}ncomplete {V}otes.
\newblock In {\em Proc. International Joint Conference on Artificial
  Intelligence (IJCAI)}, volume~9, pages 53--58, 2009.

\bibitem[BRR11]{baumeister2011computational}
Dorothea Baumeister, Magnus Roos, and J{\"o}rg Rothe.
\newblock Computational complexity of two variants of the possible winner
  problem.
\newblock In {\em Proc. 10th International Conference on Autonomous Agents and
  Multiagent Systems (AAMAS)}, pages 853--860, 2011.

\bibitem[CLMM10]{chevaleyre2010possible}
Yann Chevaleyre, J{\'e}r{\^o}me Lang, Nicolas Maudet, and J{\'e}r{\^o}me
  Monnot.
\newblock Possible winners when new candidates are added: The case of scoring
  rules.
\newblock In {\em Proc. International Conference on Artificial Intelligence
  (AAAI)}, pages 762--767, 2010.

\bibitem[DM17]{deymfcs2017}
Palash Dey and Neeldhara Misra.
\newblock On the exact amount of missing information that makes finding
  possible winners hard.
\newblock In {\em Proc. 42nd International Symposium on Mathematical
  Foundations of Computer Science}, 2017.

\bibitem[DMN15]{DeyMN15a}
Palash Dey, Neeldhara Misra, and Y.~Narahari.
\newblock Detecting possible manipulators in elections.
\newblock In {\em Proc. 2015 International Conference on Autonomous Agents and
  Multiagent Systems, {AAMAS} 2015, Istanbul, Turkey, May 4-8, 2015}, pages
  1441--1450, 2015.

\bibitem[DMN16a]{deypartial}
Palash Dey, Neeldhara Misra, and Y.~Narahari.
\newblock Complexity of manipulation with partial information in voting.
\newblock In {\em Proc. Twenty-Fifth International Joint Conference on
  Artificial Intelligence, {IJCAI} 2016, New York, NY, USA, 9-15 July 2016},
  pages 229--235, 2016.

\bibitem[DMN16b]{journalsDeyMN16}
Palash Dey, Neeldhara Misra, and Y.~Narahari.
\newblock Kernelization complexity of possible winner and coalitional
  manipulation problems in voting.
\newblock {\em Theor. Comput. Sci.}, 616:111--125, 2016.

\bibitem[DMN17]{frugaljournal}
Palash Dey, Neeldhara Misra, and Y.~Narahari.
\newblock Frugal bribery in voting.
\newblock {\em Theor. Comput. Sci.}, 676:15--32, 2017.

\bibitem[FKS03]{Fagin}
Ronald Fagin, Ravi Kumar, and D.~Sivakumar.
\newblock Efficient similarity search and classification via rank aggregation.
\newblock In {\em Proc. 2003 ACM SIGMOD International Conference on Management
  of Data}, SIGMOD '03, pages 301--312, New York, NY, USA, 2003. ACM.

\bibitem[FRRS14]{faliszewski2014complexity}
Piotr Faliszewski, Yannick Reisch, J{\"o}rg Rothe, and Lena Schend.
\newblock Complexity of manipulation, bribery, and campaign management in
  bucklin and fallback voting.
\newblock In {\em Proc. 13th International Conference on Autonomous Agents and
  Multiagent Systems (AAMAS)}, pages 1357--1358. International Foundation for
  Autonomous Agents and Multiagent Systems, 2014.

\bibitem[GJ79]{garey1979computers}
Michael~R Garey and David~S Johnson.
\newblock {\em Computers and {I}ntractability}, volume 174.
\newblock freeman New York, 1979.

\bibitem[KL05]{konczak2005voting}
Kathrin Konczak and J{\'e}r{\^o}me Lang.
\newblock Voting procedures with incomplete preferences.
\newblock In {\em Proc. International Joint Conference on Artificial
  Intelligence-05 Multidisciplinary Workshop on Advances in Preference
  Handling}, volume~20, 2005.

\bibitem[LB11]{LuB11a}
Tyler Lu and Craig Boutilier.
\newblock Robust approximation and incremental elicitation in voting protocols.
\newblock In {\em Proc. 22nd International Joint Conference on Artificial
  Intelligence (IJCAI)}, pages 287--293, 2011.

\bibitem[LPR{\etalchar{+}}07]{lang2007winner}
J{\'e}r{\^o}me Lang, Maria~Silvia Pini, Francesca Rossi, Kristen~Brent Venable,
  and Toby Walsh.
\newblock Winner determination in sequential majority voting.
\newblock In {\em Proc. 20th International Joint Conference on Artificial
  Intelligence (IJCAI)}, volume~7, pages 1372--1377, 2007.

\bibitem[LPR{\etalchar{+}}12]{lang2012winner}
J{\'e}r{\^o}me Lang, Maria~Silvia Pini, Francesca Rossi, Domenico Salvagnin,
  Kristen~Brent Venable, and Toby Walsh.
\newblock Winner determination in voting trees with incomplete preferences and
  weighted votes.
\newblock {\em Auton. Agent Multi Agent Syst.}, 25(1):130--157, 2012.

\bibitem[MBC{\etalchar{+}}16]{moulin2016handbook}
Herv{\'e} Moulin, Felix Brandt, Vincent Conitzer, Ulle Endriss, J{\'e}r{\^o}me
  Lang, and Ariel~D Procaccia.
\newblock {\em Handbook of Computational Social Choice}.
\newblock Cambridge University Press, 2016.

\bibitem[PHG00]{PennockHG00}
David~M. Pennock, Eric Horvitz, and C.~Lee Giles.
\newblock Social choice theory and recommender systems: Analysis of the
  axiomatic foundations of collaborative filtering.
\newblock In {\em Proc. Seventeenth National Conference on Artificial
  Intelligence and Twelfth Conference on on Innovative Applications of
  Artificial Intelligence, July 30 - August 3, 2000, Austin, Texas, {USA.}},
  pages 729--734, 2000.

\bibitem[PR07]{procaccia2007junta}
Ariel~D. Procaccia and Jeffrey~S. Rosenschein.
\newblock Junta distributions and the average-case complexity of manipulating
  elections.
\newblock {\em J. Artif. Intell. Res.}, 28:157--181, 2007.

\bibitem[PRVW07]{pini2007incompleteness}
Maria~Silvia Pini, Francesca Rossi, Kristen~Brent Venable, and Toby Walsh.
\newblock Incompleteness and incomparability in preference aggregation.
\newblock In {\em Proc. 20nd International Joint Conference on Artificial
  Intelligence (IJCAI)}, volume~7, pages 1464--1469, 2007.

\bibitem[Wal07]{walsh2007uncertainty}
Toby Walsh.
\newblock Uncertainty in preference elicitation and aggregation.
\newblock In {\em Proc. International Conference on Artificial Intelligence
  (AAAI)}, volume~22, pages 3--8, 2007.

\bibitem[XC11]{xia2008determining}
Lirong Xia and Vincent Conitzer.
\newblock Determining possible and necessary winners under common voting rules
  given partial orders.
\newblock {\em J. Artif. Intell. Res.}, 41(2):25--67, 2011.

\end{thebibliography}
